\documentclass[a4paper,11pt]{article}

\usepackage[utf8]{inputenc}
\usepackage[OT1]{fontenc}
\usepackage{amsmath,amssymb,amsthm,bm}
\usepackage{authblk}
\usepackage{hyperref}
\usepackage{todonotes}
\usepackage{enumitem}
\usepackage[margin=3cm]{geometry}

\newcommand{\N}{{\mathbb{N}}}
\newcommand{\R}{{\mathbb{R}}}
\newcommand{\Z}{{\mathbb{Z}}}
\newcommand{\cF}{{\mathcal{F}}}

\newcommand{\cO}{{\mathcal{O}}}
\newcommand{\cS}{{\mathcal{S}}}

\newcommand{\zero}{\mathbf{0}}
\newcommand{\one}{\mathbf{1}}
\newcommand{\la}{\langle}
\newcommand{\ra}{\rangle}
\newcommand{\lin}{\operatorname{lin}}
\newcommand{\xc}{\operatorname{xc}}
\newcommand{\sxc}{\operatorname{sxc}}
\newcommand{\conv}{\operatorname{conv}}

\newcommand{\vc}{\operatorname{VC}}
\newcommand{\lat}{\Lambda}
\newcommand{\cl}{\operatorname{cl}}

\newcommand{\psd}{\mathcal{S}_{\geq0}^d}
\newcommand{\np}{\mathrm{NP}}
\newcommand{\conp}{\mathrm{coNP}}

\theoremstyle{plain}
\newtheorem{thm}{Theorem}

\newtheorem{lem}[thm]{Lemma}

\newtheorem{ques}[thm]{Question}

\title{\textbf{Lifts for Voronoi cells of lattices}}

\author[1]{Matthias Schymura}
\author[2]{Ina Seidel}
\author[2]{Stefan Weltge}
\affil[1]{BTU Cottbus-Senftenberg, Germany}
\affil[2]{Technical University of Munich, Germany}
\date{}

\begin{document}

\maketitle

\begin{abstract}
Many polytopes arising in polyhedral combinatorics are linear projections of higher-dimensional polytopes with significantly fewer facets.
Such lifts may yield compressed representations of polytopes, which are typically used to construct small-size linear programs.
Motivated by algorithmic implications for the closest vector problem, we study lifts of Voronoi cells of lattices.

We construct an explicit $d$-dimensional lattice such that every lift of the respective Voronoi cell has $2^{\Omega(d / \log d)}$ facets.
On the positive side, we show that Voronoi cells of $d$-dimensional root lattices and their dual lattices have lifts with $\cO(d)$ and $\cO(d \log d)$ facets, respectively.
We obtain similar results for spectrahedral lifts.
\end{abstract}

\section{Introduction}
\label{sec:intro}

Many polytopes that arise in the study of polyhedral combinatorics are linear projections of higher-dimensional polytopes, also called \emph{lifts}, with significantly fewer facets.
Prominent examples include basic polytopes such as permutahedra~\cite{goemans2015smallest}, cyclic polytopes~\cite{bogomolov2015small}, and polygons~\cite{shitov2014sublinear}, as well as several polytopes associated to combinatorial optimization problems such as spanning tree polytopes~\cite{martin1991using,wong1980integer}, subtour-elimination polytopes~\cite{yannakakis1991expressing}, stable set polytopes of certain families of graphs~\cite{faenza2012separating,pulleyblank1993formulations,conforti2020extended}, matching polytopes of bounded-genus graphs~\cite{gerards1991compact}, independence polytopes of regular matroids~\cite{aprile2019regular}, or cut dominants~\cite{conforti2013extended}.

In this work, we study to which extent this phenomenon also applies to Voronoi cells of lattices.
Here, a \emph{lattice} is the image of $\Z^d$ under a linear map.
We say that a lattice is $k$-dimensional, if $k$ is the dimension of its linear hull.
The \emph{Voronoi cell} $\vc(\lat)$ of a lattice $\lat \subseteq \R^d$ is the set of all points in $\lin(\lat)$ for which the origin is among the closest lattice points, i.e.,
\[
\vc(\lat) := \left\{ x \in \lin(\lat) : \|x\| \le \|x - z\| \text{ for all } z \in \lat \right\},
\]
where $\lin(\cdot)$ denotes the linear hull and $\|\cdot\|$ denotes the Euclidean norm.
The lattice translates $z + \vc(\lat)$, $z \in \lat$, induce a facet-to-facet tiling of~$\lin(\lat)$, so that in particular Voronoi cells of lattices are what is commonly called \emph{space tiles}, see Figure~\ref{fig:lattice}.
Moreover, it is known that $\vc(\lat)$ is a centrally symmetric polytope with up to $2(2^d-1)$ facets.
We refer to~\cite[Ch.~32]{gruber2007convex} for background on translative tilings of space.

It is tempting to believe that the rich structure of Voronoi cells of lattices allows to construct polytopes with significantly fewer than $2(2^d-1)$ facets and that linearly project onto $\vc(\lat)$.
In fact, this is true for several examples: A lattice whose Voronoi cell has the largest possible number of facets is the dual root lattice~$A_d^\star$ (see Section~\ref{sec:rootLat} for a definition).
However, its Voronoi cell is a permutahedron and admits a lift with only $\cO(d \log d)$ facets~\cite{goemans2015smallest}, see Section~\ref{sec:rootLat}.
More generally, if the Voronoi cell of a $d$-dimensional lattice is a zonotope, then it has $\cO(d^2)$ generators and hence has a lift with $\cO(d^2)$ facets.
We discuss this result in detail in Section~\ref{sec:zono}.

\begin{figure}[t]
	\centering
	\includegraphics[width=0.6\textwidth]{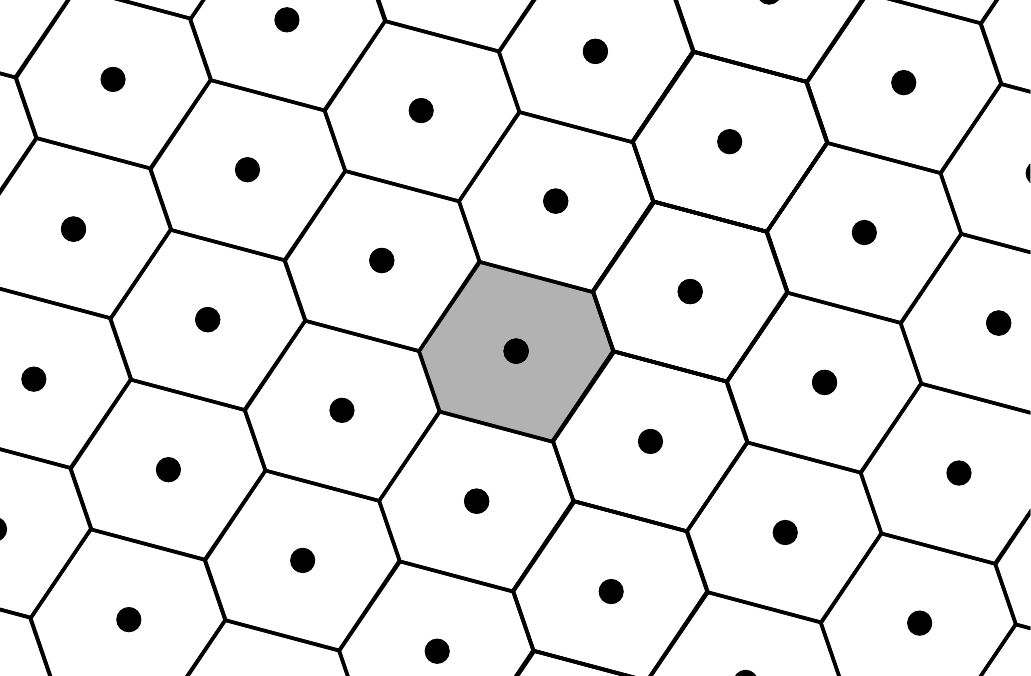}
	\caption{A lattice in $\R^2$ together with its Voronoi cell and the corresponding tiling of the plane via its lattice translates.}
	\label{fig:lattice}
\end{figure}

The lattice $A_d^\star$ also belongs to the prominent class of \emph{root lattices} and their duals.
By their algebraic and geometric properties, these lattices are prime examples in various contexts:
For example, they play a crucial role in Coxeter's classification of reflection groups (cf.~\cite[Ch.~4]{conwaysloane1999splag}), and they yield the densest sphere packings and thinnest sphere coverings in small dimensions (see~\cite{conwaysloane1999splag} or~\cite{schurmann2009computational}).

As one part of our work, we show that Voronoi cells of such lattices generally admit small lifts.
In what follows, for a polytope $P$ we write $\xc(P)$ for the minimum number of facets of any polytope that can be linearly projected onto $P$.
This number is called the \emph{extension complexity} of~$P$.

\begin{thm}
	\label{thm:xc-root-lattice}
    For every $d$-dimensional lattice $\lat$ that is a root lattice or the dual of a root lattice, we have $\xc(\vc(\lat)) = \cO(d \log d)$.
\end{thm}

This raises the question whether Voronoi cells of other lattices also have a small extension complexity, say, polynomial in their dimension.
One of the main motivations for representing a polytope $P$ as the projection of another polytope $Q$ is that a linear optimization problem over $P$ can be reduced to one over $Q$.
If $Q$ has a small number of facets, then the latter task can be expressed as a linear program with a small number of inequalities, also known as an \emph{extended formulation}.

Thus, given a lattice $\lat \subseteq \R^d$ whose Voronoi cell has a small extension complexity, we may phrase any linear optimization problem over $\vc(\lat)$ as a small-size linear program.
Such a representation may have several algorithmic consequences for the \emph{closest vector problem}.
In this problem, one is given $\lat$ and a point $x \in \R^d$ and is asked to determine a lattice point that is closest to $x$, i.e., a point in
\[
    \cl(x,\lat) := \left\{z \in \lat : \|x-z\| \le \|x-z'\| \text{ for all } z' \in \lat \right\}.
\]
Note that $z \in \cl(x,\lat)$ if and only if $x - z \in \vc(\lat)$.
Thus, a small extension complexity of $\vc(\lat)$ would yield a small-size linear program to test whether a lattice point is the closest lattice vector to $x$.
However, in view of the fact that the closest vector problem is $\np$-hard~\cite{van1981another} and the belief that $\np \ne \conp$, we do not expect efficient algorithms that, for general lattices (given in form of a basis), decide whether a point is the closest lattice vector to $x$.
Another sequence of algorithmic implications arises from the algorithm of Micciancio \& Voulgaris~\cite{micciancio2013deterministic}, which also motivated other recent work on compact representations of Voronoi cells, such as~\cite{hunkenschroederreulandschymura2020compact}, see also~\cite[\S~3.7]{hunkenschroderphd}.

We remark that the mere \emph{existence} of small size extended formulations of Voronoi cells may not be immediately applicable, since finding such representations as well as verifying that they indeed yield the Voronoi cell of a given lattice might be hard.
Thus, polynomial bounds on the extension complexities of Voronoi cells of general lattices would not contradict hardness assumptions in complexity theory.
In fact, we initially considered the possibility of such bounds.

However, as our main result we explicitly construct lattices with Voronoi cells of extension complexity close to the trivial upper bound $2(2^d-1)$.

\begin{thm} \label{thmLowerBound}
    There exists a family of $d$-dimensional lattices $\lat$ such that $\xc(\vc(\lat)) = 2^{\Omega(d / \log d)}$.
\end{thm}

Lower bounds on extension complexities have been established for various prominent polytopes in recent years.
Of particular note are results for cut polytopes~\cite{fiorini2015exponential,kaibel2015short,chan2016approximate}, matching polytopes~\cite{rothvoss2017matching}, and certain stable set polytopes~\cite{GJW}.
Lower bounds for other polytopes $Q$ are typically obtained by showing that a face $ F $ of $Q$ affinely projects onto one of the polytopes~$P$ from above and using the simple fact $\xc(P) \le \xc(F) \le \xc(Q)$. 
Unfortunately, it seems difficult to construct lattices for which this approach can be directly applied to the Voronoi cell.
However, we will exploit the lesser known fact that $\xc(Q) = \xc(Q^\circ)$ holds for every polytope $Q$ with the origin in its interior, where $Q^\circ$ is the dual polytope of $Q$.
In fact, we will describe a way to obtain many $0/1$-polytopes as projections of faces of dual polytopes of Voronoi cells of lattices.
As an example, for every $n$-node graph~$G$ we can construct a lattice $\lat$ of dimension at most $n+1$ such that the stable set polytope of $G$ is a projection of a face of $\vc(\lat)^\circ$.
Theorem~\ref{thmLowerBound} then follows from a construction of Göös, Jain \& Watson~\cite{GJW} of stable set polytopes with high extension complexity. 

Another prominent way of representing polytopes is via linear projections of feasible regions of \emph{semidefinite} programs, i.e., spectrahedra.
We will discuss how our approach also yields a version of Theorem~\ref{thmLowerBound} for such semidefinite lifts with a slightly weaker but still superpolynomial bound.

\paragraph{Outline} 
In Section~\ref{sec:prel}, we provide a brief introduction to lifts of polytopes and lattices, focusing on tools and properties that are  essential for our arguments in the following sections. 
In Section~\ref{sec:upperBound}, we derive upper bounds on the extension complexity of Voronoi cells for some selected classes of lattices, such as root lattices and their duals, zonotopal lattices, and a class of lattices that do not admit a compact representation in the sense of~\cite{hunkenschroederreulandschymura2020compact}.
The proof of Theorem~\ref{thmLowerBound} is given in Section~\ref{sec:lower}, and in Section~\ref{sec:semidef}, we briefly introduce semidefinite lifts and present a version of Theorem~\ref{thmLowerBound} with a superpolynomial bound on the semidefinite extension complexity.
We close our paper with a discussion of open problems in Section~\ref{sec:openquestions}.

\section{Preliminaries}
\label{sec:prel}

\subsection{Extension complexity: A toolbox}

Throughout this paper we only need basic facts regarding extension complexities of polytopes and most of them are well-known. For the sake of completeness, we provide proofs here. 
First, we start with a simple fact already mentioned in the introduction.

\begin{lem}
	\label{lem:xcFace}
	For every face $F$ of a polytope $P$, we have $\xc(F) \le \xc(P)$.
\end{lem}

\begin{proof}
	If $P$ is the image of a polyhedron $Q$ with $k$ facets under a linear map~$\tau$, then $F$ is the image of $ \tau^{-1}(F) \cap Q $, which is a face of $Q$ and hence has at most $k$ facets.
\end{proof}

For the next fact we need the notion of a slack matrix of a polytope. 
To this end, we consider a polytope $ P = \{ x \in \R^d : \la a_i , x \ra \le b_i, i \in [m] \} = \conv\{v_1, \dots, v_n\} $, where $[m] := \{1, \dots, m\}$ and $ \la \cdot, \cdot \ra $ denotes the standard Euclidean scalar product. Corresponding to these two descriptions of $ P $, we define the \emph{slack matrix} $ S = (S_{i,j}) \in \R_{\ge 0}^{m \times n} $ via $ S_{i,j} = b_i - \la a_i , v_j \ra $. 
Yannakakis~\cite{yannakakis1991expressing} showed that the extension complexity $ \xc(P) $ of $ P $ equals the \emph{nonnegative rank} of $ S $, which is the smallest number~$ r $ such that $ S = FV$, where $ F \in \R_{\ge 0}^{m \times r} $ and $ V \in \R_{\ge 0}^{r \times n} $, and which is denoted by $ r_+(S) $. 

For a polytope $P$ containing the origin $ \zero $ in its relative interior, the \emph{dual polytope} of $P$ is defined as
\[
P^\circ := \{y \in \lin(P) : \la x,y \ra \le 1 \text{ for all } x \in P\}.
\]
It is a basic fact that $P^\circ$ is again a polytope with the origin in its relative interior, $\lin(P^\circ) = \lin(P)$, and $(P^\circ)^\circ = P$.
Moreover, it is easy to see that if 
\[
P = \{ x \in \lin(P) : \la w_i , x \ra \le 1, i \in [m] \} = \conv\{v_1, \dots, v_n\},
\]
then 
\begin{equation}
	\label{eq:dualPolytope}
	P^\circ = \{y \in \lin(P) : \la v_i , x \ra \le 1, i \in [n]\} = \conv \{ w_1, \dots, w_m \}.
\end{equation} 
In particular, this shows that if $S$ is a slack matrix of $P$ induced by $ v_1, \dots, v_n$ and $w_1, \dots, w_m$, then $S^\intercal$ is a slack matrix of $P^\circ$.
Since $r_+(S) = r_+(S^\intercal)$ we obtain the following fact.

\begin{lem}
\label{lem:xc-dual}
	For every polytope $ P \subseteq \R^d $ that contains the origin in its relative interior, we have
	\[
		\xc(P) = \xc(P^\circ).
	\]
\end{lem}
\noindent The next statement shows that the extension complexity behaves well under Cartesian products, Minkowski sums and intersections.

\begin{lem}
\label{lem:xc-simple-constructions}
If $P \subseteq \R^d$, $Q \subseteq \R^{d'}$ are polytopes, then
\begin{enumerate}[label=(\roman*)]
 \item[(i)] $\xc(P \times Q) \leq \xc(P) + \xc(Q)$.
\end{enumerate}
Moreover, if $d = d'$, then
\begin{enumerate}
 \item[(ii)] $\xc(P + Q) \leq \xc(P) + \xc(Q)$ and
 \item[(iii)]\label{lbl:xc-simple-constructions-intersection} $\xc(P \cap Q) \leq \xc(P) + \xc(Q)$.
\end{enumerate}
\end{lem}

\begin{proof}
	(i): If $P'$ linearly projects onto $P$ and $Q'$ onto $Q$, then $P' \times Q'$ linearly projects onto $P \times Q$. Moreover, the number of facets of $P' \times Q'$ is equal to the sum of the number of facets of $P'$ and $Q'$. \\
	(ii): The polytope $P \times Q$ linearly projects onto $P+Q$ via $(p,q) \mapsto p+q$ for $(p,q) \in P \times Q$, and hence the claim follows from (i).\\
	(iii): If $P=\pi(P')$ and $Q = \tau(Q')$ hold for some polyhedra $P',Q'$ and linear maps $\pi,\tau$, then $P \cap Q$ is a linear image of the polyhedron $L = \{ (y,z) \in P' \times Q' : \pi(y) = \tau(z) \}$. Moreover, the number of facets of $L$ is at most the number of facets of $P' \times Q'$, which, again, is equal to the sum of the number of facets of $P'$ and $Q'$. 
\end{proof}

The next fact is a very useful result following from a work of Balas~\cite{balas1979disjunctive} deriving a description of the convex hull of the union of certain polytopes. The proof of the version presented here can be found in~\cite[Prop.~3.1.1]{weltge2015diss}.

 \begin{lem} 
 	\label{lem:xc-union}
 	Let $ P_1, \dots, P_k $ be polytopes, then
 	\[
 		\xc(\conv(P_1 \cup \ldots \cup P_k)) \leq \sum_{i=1}^k \xc(P_i) + \left|\left\{ i \in [k] : \dim(P_i) = 0 \right\}\right|.
 	\]
 \end{lem}

\noindent We mentioned already that some lattices have a permutahedron as their Voronoi cell.
These polytopes arise from a single vector by permuting its coordinates in all possible ways and taking their convex hull.
Let us denote the set of all bijective maps on $[d]$ by $S_d$.
For a permutation $\pi \in S_d$ and a vector $v = (v(1), \dots, v(d)) \in \R^d$, let $ \pi(v) := (v(\pi(1)),\dots,v(\pi(d))) $ be the vector that arises from $v$ via permuting its entries according to $\pi$.

\begin{lem}
	\label{lem:xc-simple-constructions-perms}
	For every $v \in \R^d $ we have $ \xc(\conv\{ \pi(v) : \pi \in S_d \}) \le d^2$.
\end{lem}

\begin{proof}
For $\pi \in S_d$, let $ P(\pi) \in \{0,1\}^{d\times d}$ with $P(\pi)_{ij} = 1$ if and only if $ \pi(i) = j $, for all $i,j \in [d]$, be the associated permutation matrix.
It is easy to see that $\conv\{ \pi(v) : \pi \in S_d \}$ is the image of $ B_d := \conv\{ P(\pi)  : \pi \in S_d \} $ under the linear map $ \tau: \R^{d\times d} \rightarrow \R^d $ with $\tau(X)_i = \sum_{j=1}^d v_jX_{ij}$, for $i \in [d]$.
The latter polytope is the \emph{Birkhoff–von Neumann polytope} \cite{birkhoff1946tres, von1953certain} described via 
\begin{align*}
	B_d =
	\left\{X \in \R_{\ge 0}^{d \times d} : \sum_{i=1}^d X_{ij} = 1 \text{ for } j \in [d], \,
	\sum_{j=1}^d X_{ij} = 1 \text{ for } i \in [d] 
	\right\},
\end{align*}
which has $d^2$ facets. 
\end{proof}

Goemans~\cite{goemans2015smallest} showed that if $v = (1,2,\dots,d)$, then the above bound can be improved to $ \xc(\conv\{ \pi(v) : \pi \in S_d \}) = \Theta(d \log d)$.

\subsection{Lattices and Voronoi cells}

Most basic notions regarding lattices and their Voronoi cells have been already introduced in Section~\ref{sec:intro}. In this section, we provide some further definitions and results that we use to obtain bounds on the extension complexity of Voronoi cells of lattices.

We call two lattices $ \lat, \Gamma \subseteq \R^d$ \emph{isomorphic}, if there exists an orthogonal matrix $Q \in \R^{d \times d}$ such that $Q\lat = \Gamma$. Note that $\vc(\Gamma) = Q\vc(\lat) $ and therefore the extension complexities of their Voronoi cells coincide. 

In some parts, we will consider the \emph{dual lattice} of a lattice $\lat \subseteq \R^d$, which is defined as
\[
\lat^\star = \left\{ x \in \lin(\lat) : \la x,y \ra \in \Z \textrm{ for all } y \in \lat\right\}.
\]
Note that for every two lattices $\Lambda, \Gamma$, their product $ \Lambda \times \Gamma $ is also a lattice.
The following lemma shows that the Cartesian product behaves well with respect to Voronoi cells or duals of lattices.

\begin{lem}
	\label{lem:vc-lat-dual-product}
	For any two lattices $\Lambda \subseteq \R^d$ and $\Gamma \subseteq \R^{d'}$ we have
	\begin{enumerate}[label=(\roman*)]
		\item $\vc(\Lambda \times \Gamma) = \vc(\Lambda) \times \vc(\Gamma)$, and
		\item $(\Lambda \times \Gamma)^\star = \Lambda^\star \times \Gamma^\star$.
	\end{enumerate}
\end{lem}
The proof is straightforward from the definitions and is left as an exercise.
A main ingredient for proving Theorem~\ref{thmLowerBound} is to consider the dual polytope $\vc(\lat)^\circ$ of $\vc(\lat)$.
Recall that we have $\xc(\vc(\lat)) = \xc(\vc(\lat)^\circ)$ by Lemma~\ref{lem:xc-dual}.
The following two observations are crucial for our arguments.

\begin{lem}
	\label{lem:dualVC}
	For every lattice $\Lambda$ we have
	\begin{equation*}
		\vc(\lat)^\circ = \conv \left \{ \tfrac{2}{\|z\|^2} z : z \in \lat \setminus \{\zero\} \right \}.
	\end{equation*}
\end{lem}

\begin{proof}
	In view of the identities
	\begin{align*}
		\vc(\lat) = & \left \{ x \in \lin(\lat) : \|x\|^2 \le \|x - z\|^2 \text{ for all } z \in \lat \right \} \\
		= & \left \{ x \in \lin(\lat) : \la x,z \ra \le \tfrac{1}{2} \|z\|^2 \text{ for all } z \in \lat \right \} \\
		= & \left \{ x \in \lin(\lat) : \la x,\tfrac{2}{\|z\|^2} z \ra \le 1 \text{ for all } z \in \lat \setminus \{\zero\} \right \},
	\end{align*}
	the claim follows from \eqref{eq:dualPolytope}.
\end{proof}

\begin{lem}
	\label{lem:polarFace}
	Let $\Lambda \subseteq \R^d$ be a lattice and $p \in \R^d$. If $\zero \in \cl(p,\lat)$, then 
	\[
	\conv \left \{ \tfrac{2}{\|z\|^2} z : z \in \cl(p,\lat) \setminus \{\zero\} \right \}
	\]
	is a face of $\vc(\lat)^\circ$.
\end{lem}
\begin{proof}
	Since $\zero \in \cl(p,\lat)$, every nonzero lattice point $z \in \lat \setminus \{\zero\}$ satisfies $ \|p - z\|^2 \ge \|p\|^2 $,
	with equality if and only if $z \in \cl(p,\lat) \setminus \{\zero\}$.
	Note that the above inequality is equivalent to $ \la p,\tfrac{2}{\|z\|^2}z \ra \le 1 $.
	Thus, due to Lemma~\ref{lem:dualVC} we see that $ F := \{ y \in \vc(\lat)^\circ : \la p,y \ra = 1 \} $ is a face of $\vc(\lat)^\circ$.
	This establishes the claim since
	\begin{align*}
		F &= \conv \left \{ \tfrac{2}{\|z\|^2} z : z \in \lat \setminus \{\zero\}, \, \la p,\tfrac{2}{\|z\|^2}z \ra = 1 \right \} \\
		&= \conv \left \{ \tfrac{2}{\|z\|^2} z : z \in \cl(p,\lat) \setminus \{\zero\} \right \}. \qedhere
	\end{align*}
\end{proof}

\section{Lattices with small extension complexity}
\label{sec:upperBound}

In this section, we provide bounds on the extension complexities of Voronoi cells of some prominent lattices.

\subsection{Root lattices and their duals}
\label{sec:rootLat}

We start with Voronoi cells of root lattices and their duals.
An \emph{irreducible root lattice} is a lattice $\lat$ for which there exists a finite set~$S$ of vectors of squared length equal to~$1$ or~$2$, such that $\lat = \{ \sum_{b \in S} \alpha_b b : \alpha_b \in \Z \text{ for all } b \in S\}$.
We say that a lattice is a \emph{(general) root lattice}, if it is isomorphic to a lattice obtained by iteratively taking Cartesian products with irreducible root lattices.
A well-known theorem related to the classification of reflection groups states that besides the lattice $\Z^d$ of integers, up to isomorphism the irreducible root lattices split into the two infinite classes
\begin{align*}
	A_d &= \left\{ x \in \Z^{d+1} : x(1) + \ldots + x(d+1) = 0 \right\} \quad \textrm{and} \\
	D_d &= \left\{ x \in \Z^d : x(1) + \ldots + x(d) \textrm{ is even} \right\} ,
\end{align*}
and the three exceptional lattices
\begin{align*}
	E_8 &= D_8 \cup \left( \tfrac12 \one + D_8 \right) , \\
	E_7 &= \left\{ x \in E_8 : \la x , e_7 + e_8 \ra = 0 \right\} \textrm{ and} \\ 
	E_6 &= \left\{ x \in E_7 : \la x , e_6 + e_8 \ra = 0 \right\}.
\end{align*}
Here and in the following, we denote by $e_i$ the $i$th standard Euclidean unit vector and by~$\one$ the all-one vector in the corresponding space.
Moreover, the dual lattices of the two infinite classes $A_d$ and $D_d$ are given by
\[
A_d^\star = \bigcup_{i=0}^d \left( v_i + A_d \right),
\]
with $v_i = \Big( \underbrace{\tfrac{i}{d+1}, \ldots, \tfrac{i}{d+1}}_{j \textrm{ times}}, \underbrace{-\tfrac{j}{d+1}, \ldots, -\tfrac{j}{d+1}}_{i \textrm{ times}}\Big)$, for $0 \leq i \leq d$ and $j=d+1-i$, and
\[
D_d^\star = \Z^d \cup (\tfrac12\one + \Z^d),
\]
respectively.
In the literature the dual $D_d^\star$ is usually scaled by a factor of~$2$ in order to get an integral lattice, which is often more convenient to investigate.
In order to avoid confusion, we denote it by
\[
\bar D_d^\star := 2D_d^\star = (2\Z^d) \cup (\one + 2\Z^d),
\]
and note that this scaling has no effect on the extension complexity of its Voronoi cell.
We refer to Conway \& Sloane~\cite[Ch.~4 \& Ch.~21]{conwaysloane1999splag} and Martinet~\cite[Ch.~4]{martinet2003perfect} for proofs, original references and background information on root lattices.
Details on Voronoi cells and Delaunay polytopes of root lattices can be found in Moody \& Patera~\cite{moodypatera1992voronoi}, which together with the two aforementioned monographs are our main sources of information.

Given a lattice $\lat \subseteq \R^d$ we write $|\lat| = \min\{ \|z\| : z \in \lat \setminus \{\zero\}\}$ for the length of a shortest non-trivial vector in~$\lat$.
A \emph{minimal vector} of~$\lat$ is any vector $z \in \lat$ with $\|z\| = |\lat|$, and a \emph{facet vector} of~$\lat$ is any vector $w \in \lat$, such that the constraint $\la x,w \ra \leq \frac12 \|w\|^2$ defines a facet of the Voronoi cell~$\vc(\lat)$.
For convenience, we write
\begin{align*}
	\cS(\lat) &= \left\{ z \in \lat : \|z\| = |\lat| \right\} \quad \text{and} \\
	\cF(\lat) &= \left\{ w \in \lat : w \text{ is a facet vector of } \lat \right\},
\end{align*}
for the set of minimal vectors and facet vectors, respectively.
In general, one has the inclusion $\cS(\lat) \subseteq \cF(\lat)$, which however is usually strict.
Root lattices are now neatly characterized by the property that every facet vector is at the same time a minimal vector, that is, the equality $\cS(\lat) = \cF(\lat)$ holds (see Rajan \& Shende~\cite{rajanshende1996acharacterization}).

Since the set of minimal vectors of the irreducible root lattices are well-understood, this allows to describe their Voronoi cells as well.
For the sake of the asymptotic study of the extension complexity of their Voronoi cells, it suffices to understand the two infinite families~$A_d$ and $D_d$, and their duals~$A_d^\star$ and $D_d^\star$.
In the sequel, we provide bounds on the extension complexities of the Voronoi cells of these lattices. 
To achieve these bounds, we sometimes use a characterization of the facet vectors and in other cases we use a characterization of the vertices of the Voronoi cell.
For the sake of easy reference, we describe the vertices and facet vectors in all cases.
Due to Lemma~\ref{lem:xc-simple-constructions} and Lemma~\ref{lem:vc-lat-dual-product}, these bounds directly imply Theorem~\ref{thm:xc-root-lattice}. Moreover, the bound in Theorem~\ref{thm:xc-root-lattice} is asymptotically tight since the Voronoi cell of $A_d^\star$ is a permutahedron, see Lemma~\ref{lem:xcAstar}.

\subsubsection{Voronoi cell of \texorpdfstring{$A_d$}{A\_d}}

The Voronoi cell of the root lattice $A_d$ is given by
\[
\vc(A_d) = \conv \left\{ \pi(v_i) : \pi \in S_{d+1}, i \in \{0, \dots, d\}\right\},
\]
where
\[
v_i = \Big( \underbrace{\tfrac{i}{d+1}, \dots, \tfrac{i}{d+1}}_{j \text{ times}}, \underbrace{-\tfrac{j}{d+1}, \dots, -\tfrac{j}{d+1}}_{i \textrm{ times}} \Big) \in \R^{d+1}
\]
with $j = d+1-i$.
Moreover, we have
\begin{align*}
\vc(A_d) &= \left \{ x \in \R^{d+1} : \la x,z \ra \le 1 \text{ for all } z \in \cF(A_d) \right \}, \text{ where }\\
\cF(A_d) &= \left\{ \pi((1,-1,0, \dots, 0)) : \pi \in S_{d+1} \right\}
\end{align*}
(see~\cite[Ch.~21 \& Ch.~4, Sec.~6]{conwaysloane1999splag}).

\goodbreak
\begin{lem} \label{lem:xcA}
	$\xc(\vc(A_d)) = \cO(d)$.
\end{lem}

\begin{proof}
Using the description of the facet vectors stated above, we obtain that $\vc(A_d)^\circ = S + (-S)$, where $S$ is the $d$-dimensional simplex $S = \conv\{e_1, \dots, e_{d+1}\}$.
Hence, using Lemma~\ref{lem:xc-simple-constructions} and Lemma~\ref{lem:dualVC}, we obtain the upper bound $\xc(\vc(A_d)) \leq 2(d+1)$.
\end{proof}

\subsubsection{Voronoi cell of \texorpdfstring{$D_d$}{D\_d}}

The Voronoi cell of $D_d$ is given by
\begin{align*}
\vc(D_d) = \conv \left(\left\{ \pm e_1,\dots, \pm e_d \right\} \cup \{-\tfrac12,\tfrac12\}^d \right).
\end{align*}
Moreover, we have
\begin{align*}
\vc(D_d) &= \left \{ x \in \R^{d+1} : \la x,z \ra \le 1 \text{ for all } z \in \cF(D_d) \right \}, \text{ where }\\
\cF(D_d) &= \left\{ \pm e_i \pm e_j : 1 \leq i < j \leq d \right\}.
\end{align*}
This follows from the characterization of the minimal (and thus facet) vectors of $D_d$ given in~\cite[Ch.~4, Sec.~7]{conwaysloane1999splag}.
The inner description of the Voronoi cell can be read off from the vertices of a fundamental simplex for~$D_d$ (see~\cite[Ch.~21, Fig.~21.7]{conwaysloane1999splag}). 

\begin{lem} \label{lem:xcD}
	$\xc(\vc(D_d)) = \cO(d)$.
\end{lem}

\begin{proof}
	Using the description of the vertices of $\vc(D_d)$ stated above, we obtain that 
	\[
	\vc(D_d)^\circ = 2 \cdot \conv\{ \pm e_1, \dots, \pm e_d\} \cap [-1,1]^d.
	\]
	Hence, the dual of the Voronoi cell is the intersection of a hypercube and a crosspolytope.
	Since
	\begin{equation}
		\label{eq:xcCube}
		\xc([-1,1]^d) = \xc(\conv\{\pm e_1, \dots, \pm e_d\}) = 2d,
	\end{equation}
	(see, e.g.,~\cite[Cor.~2.5]{grandepadrolsanyal2018extension}), Lemmas~\ref{lem:xc-simple-constructions} and~\ref{lem:dualVC} imply $\xc(\vc(D_d)) = \cO(d)$.
\end{proof}

\subsubsection{Voronoi cell of \texorpdfstring{$A_d^\star$}{A\_d*}}
\label{sec:AnStar}

The Voronoi cell of the dual of the root lattice $A_d^\star$ is given by
\[
\vc(A_d^\star) = \conv \left\{ \pi(v) : \pi \in S_{d+1} \right\},
\]
where
\[
v = \tfrac{1}{2d+2}\left(-d,-d+2,\ldots,d-2,d\right) \in \R^{d+1}.
\]
Moreover, we have
\begin{align*}
\cF(A_d^\star) = \left\{ v \in \lin(A_d^\star) : v \text{ is a vertex of } \vc(A_d) \right\}.
\end{align*}
The characterization of the vertices can be found in~\cite[Ch.~21, Sec.~3F]{conwaysloane1999splag} and the fact that the facet vectors are exactly the vertices of $\vc(A_d)$ is explained in detail in the unpublished monograph~\cite[Ch.~3.5]{engelmichelsenechal2004lattice}.

\begin{lem} \label{lem:xcAstar}
	$\xc(\vc(A_d^\star)) = \Theta(d \log{d})$.
\end{lem}

\begin{proof}
	Using the description of the vertices of $\vc(A_d^\star)$ stated before, we obtain that $\vc(A_d^\star)$ is an affine linear transformation of the standard permutahedron
	\[
		P_d = \left\{ (\pi(1),\ldots,\pi(d+1)) : \pi \in S_{d+1} \right\}.
	\]
	In fact, 
	\[
	\vc(A_d^\star) = \tfrac{1}{d+1} P_d - \tfrac{d+2}{2d+2} \one.
	\]
	The claim follows, since Goemans~\cite{goemans2015smallest} showed that the extension complexity of~$P_d$ is in $\Theta(d \log{d})$.
\end{proof}

\subsubsection{Voronoi cell of \texorpdfstring{$D_d^\star$}{D\_d*}}

As explained before, we consider the integral lattice $\bar D_d^\star$ instead of~$D_d^\star$.
The Voronoi cell of $\bar D_d^\star$ is given by
\[
\vc(\bar D_d^\star) = \conv \left\{ \pi(v) : \pi \in S_d, v \in V \right\}, 
\]
where
\[
V = \begin{cases}
\{0\}^{\tfrac{d}{2}} \times \{-1,1\}^{\tfrac{d}{2}} & \text{, if } d \text{ is even,} \\
\{0\}^{\tfrac{d-1}{2}} \times \{-\tfrac12,\tfrac12\} \times \{-1,1\}^{\tfrac{d-1}{2}} & \text{, if } d \text{ is odd.}
\end{cases}
\]
Moreover, we have
\begin{align*}
\cF(\bar D_d^\star) = \left\{\pm 2 e_1, \dots, \pm 2 e_d\right\} \cup \{-1,1\}^d.
\end{align*}
We refer to~\cite[Ch.~21, Sect.~3E]{conwaysloane1999splag} for the characterization of the facet vectors and the inner description of the Voronoi cell, which is therein denoted by the symbols $\beta(d,d/2)$, for $d$ even, and $\tfrac12 \delta(d,(d-1)/2)$, for $d$ odd.

\begin{lem}
	\label{lem:xcDstar}
	$\xc(\vc(D_d^\star)) = \cO(d)$.
\end{lem}

\begin{proof}
	Using the above description of the facet vectors, we obtain that 
	\[
	\vc(\bar D_d^\star) = [-1,1]^d \cap \tfrac{d}{2} \cdot \conv\{\pm e_1, \dots, \pm e_d\}.
	\]
	Hence, the Voronoi cell of $D_d^\star$ is the intersection of a hypercube and a crosspolytope.
	As in the case of the root lattice $D_d$, the stated bound follows by Lemma~\ref{lem:xc-simple-constructions} and Equation~\eqref{eq:xcCube}.
\end{proof}

Note that all the bounds stated in Lemmas~\ref{lem:xcA},~\ref{lem:xcD}, and~\ref{lem:xcDstar} are asymptotically tight, since the extension complexity of a polytope grows at least linearly with its dimension (cf.~\cite[Eq.~2 \& Prop.~5.2]{fiorinikaibelpashkovichtheis2013combinatorial}).

\subsection{Zonotopal lattices}
\label{sec:zono}

A \emph{zonotope} $Z \subseteq \R^d$ is the Minkowski sum of finitely many line segments, that is, there are vectors $a_1,b_1,\ldots,a_m,b_m \in \R^d$ such that $Z = \sum_{i=1}^m \conv\{a_i,b_i\}$.
The non-zero vectors $z_i = b_i - a_i$ are usually called the \emph{generators} of the zonotope, and clearly, $Z$ is an affine projection of the $m$-dimensional cube $[-1,1]^m$ via $e_i \mapsto z_i$, for $1 \leq i \leq m$, and a suitable translation.
Regarding the extension complexity of a zonotope~$Z$, the bound $\xc(Z) \leq 2m$ thus immediately follows from the definition.

A lattice $\Lambda \subseteq \R^d$ is said to be \emph{zonotopal} if its Voronoi cell is a zonotope.
Every lattice of dimension at most three is zonotopal, but from dimension four on there exist non-zonotopal lattices.
For instance, the Voronoi cell of the root lattice $D_4$ is the non-zonotopal $24$-cell.
Examples of classes of zonotopal lattices are $\Z^d$, the root lattice~$A_d$, its dual lattice~$A_d^\star$, lattices of Voronoi's first kind, and the tensor product $A_d \otimes A_{d'}$.
Zonotopal space tiles have been extensively studied over the years, mostly due to their combinatorial connections to regular matroids, hyperplane arrangements, and totally unimodular matrices.
For a detailed account on zonotopal lattices and pointers to the original works containing the previous statements we refer to~\cite[Sect.~2]{mccormickpeisscheidweilervallentin2020apolynomial}.

The tiling constraint on a zonotope that arises as the Voronoi cell of a lattice, allows it to have at most quadratically many generators in terms of its dimension.
In particular, these polytopes admit lifts with quadratically many facets.

\begin{thm}
	\label{thm:zonotopal-lattices}
	Each zonotopal lattice $\lat \subseteq \R^d$ satisfies
	$\xc(\vc(\lat)) \leq d(d+1)$.
\end{thm}

\begin{proof}
	It suffices to argue that the Voronoi cell is generated by at most~$\binom{d+1}{2}$ line segments.
	Indeed, each line segment $L$ satisfies $\xc(L) = 2$ and hence the statement follows using Lemma~\ref{lem:xc-union}.

	Erdahl~\cite[Sect.~5]{erdahl1999zonotopes} proved that the generators of a space tiling zonotope correspond to the normal vectors of a certain dicing.
	A \emph{dicing} in~$\R^d$ is an arrangement of hyperplanes consisting of~$r \geq d$ families of infinitely many equallyspaced hyperplanes such that: (1) there are~$d$ families whose corresponding normal vectors are linearly independent, and (2) every vertex of the arrangement is contained in a hyperplane of each family.
	
	By~\cite[Thm.~3.3]{erdahl1999zonotopes}, every dicing is affinely equivalent to a dicing whose set of hyperplane normal vectors -- one normal vector for each of the~$r$ families -- consists of the columns of a totally unimodular $d \times r$ matrix.
	By construction, this totally unimodular matrix is such that for any two of its columns $v,w$, we have $v \neq \pm w$ and $v,w \neq \zero$.
	A classical result, that is often attributed to Heller~\cite{heller1957onlinear} but already appears in Korkine \& Zolotarev~\cite{korkinezolotarev1877sureles}, yields that every such totally unimodular $d \times r$ matrix has at most $r \leq \binom{d+1}{2}$ columns.
	Thus, the zonotopal Voronoi cell $\vc(\lat)$ is generated by at most~$\binom{d+1}{2}$ line segments.
\end{proof}

	Alternatively, the fact that zonotopal Voronoi cells in~$\R^d$ are generated by at most~$\binom{d+1}{2}$ line segments also follows from Voronoi's reduction theory.
	The Delaunay subdivisions of zonotopal lattices correspond to certain polyhedral cones (Voronoi's L-types) in the cone~$\psd$ of positive semi-definite $d \times d$ matrices that are generated by rank one matrices.
	Since~$\psd$ has dimension $\binom{d+1}{2}$, Carath\'{e}odory's Theorem yields the bound.
	We refer the reader to Erdahl~\cite[Sect.~7]{erdahl1999zonotopes} for an intuitive description and references to the original works.

\subsection{Lattices defined by simple congruences}
\label{sec:congruence}

For any $a \in \N$, we consider the lattice
\begin{align}
	\Lambda_d(a) := \left\{ x \in \Z^d : x_1 \equiv x_2 \equiv \ldots \equiv x_d \!\!\mod a \right\}.\label{eq:def-congruence-lattice}
\end{align}
The case $a = \lceil \tfrac{d}{2} \rceil$ played a special role in~\cite[Thm.~2]{hunkenschroederreulandschymura2020compact} for the determination of lattices that do not have a basis that admits a compact representation of the Voronoi cell.
To this end, the authors determined the set $\cF(\Lambda_d(\lceil \tfrac{d}{2} \rceil))$ of facet vectors explicitly (there are exponentially many of them).
However, their proof can be extended to general $a$ to give a description of the facet vectors of $\cF(\Lambda_d(a))$ that is precise enough to allow drawing conclusions towards small extended formulations.

\begin{lem}
	\label{lem:facet-vectors-congruence-lattices}
	For all $a \in \N$, the set of facet vectors of $\Lambda_d(a)$ is contained in 
	\begin{align*}
		\cF(\Lambda_d(a)) \subseteq & \ \{\one, -\one \} \cup \{\pm a e_i : i \in [d]\}  \\
		& \ \cup \left\{ v_{S,\ell} \in \R^d : \emptyset \neq S \subsetneq [d], \ell \in \left\{ \left\lfloor \tfrac{a|S|}{d} \right\rfloor , \left\lceil \tfrac{a|S|}{d} \right\rceil \right\} \right\},
	\end{align*}
	where $v_{S,\ell}(i) = a - \ell$, if $i \in S$, and $v_{S,\ell}(i) = -\ell $, if $i \notin S$.
\end{lem}

\begin{proof}
	Follows directly with the proof of~\cite[Lem.~3]{hunkenschroederreulandschymura2020compact}.
\end{proof}

\begin{thm}
	For all $a \in \N$, we have
	$
	\xc(\vc(\Lambda_d(a))) \in \cO(d^3)
	$.
\end{thm}

\begin{proof}
	Due to Lemma~\ref{lem:facet-vectors-congruence-lattices} and Equation~\eqref{eq:dualPolytope}, the dual polytope of the Voronoi cell of $ \Lambda_d(a) $ equals 
	\[
		\vc(\Lambda_d(a))^\circ = \conv\left( V_{\pm \one} \cup V_{\pm a} \cup \bigcup_{k, \ell} V_{k,\ell}   \right),
	\]
	where the last union is over all $ k \in [d-1], $ $ \ell \in \{ \lfloor \frac{ak}{d} \rfloor , \lceil \frac{ak}{d} \rceil \}$ and the sets $V_{\pm \one}, V_{\pm a}$ and $V_{k,\ell}$ are defined as follows:
	\begin{align*} 
		V_{\pm \one} &:= \conv\left\{\tfrac{2}{d} \one, -\tfrac{2}{d} \one\right\} , \\
		V_{\pm a} &:= \conv \left\{\pm\tfrac{2}{a^2} a e_i: i \in [d]\right\} \text{ and } \\
		V_{k,\ell} &:= \conv\Big\{ \tfrac{2}{k(a-\ell)^2 + (d-k)\ell^2} z : z \in \{a-\ell,-\ell\}^d \\
		&\phantom{:= \conv\Big\{ \tfrac{2}{k(a-\ell)^2 + (d-k)\ell^2} z : \ } \text{ with exactly } k \text{ entries equal to } a-\ell   \Big\}.
	\end{align*}
	Clearly, $\xc(V_{\pm \one}) = 2$. Moreover, $ \xc(V_{\pm a}) = 2d $, since $ V_{\pm a} $ is a crosspolytope, see~\eqref{eq:xcCube}.
	Furthermore, for $ k \in [d-1] $ and $ \ell \in \{ \lfloor \tfrac{ak}{d} \rfloor , \lceil \tfrac{ak}{d} \rceil \}$, using Lemma~\ref{lem:xc-simple-constructions-perms} and the fact that $V_{k,\ell}$ equals
	\[
	V_{k,\ell} = \conv\{\pi(v_{k,\ell}) : \pi \in S_d\},
	\]
	where $ v_{k,\ell} = (a-\ell, \dots, a-\ell, -\ell, \dots, -\ell) $ with exactly $ k $ entries equal to $a-\ell$, we obtain $\xc(V_{k,\ell}) \le d^2 $.
	
	Combining these bounds and applying Lemma~\ref{lem:xc-union} and Lemma~\ref{lem:xc-dual}, we obtain the desired bound.
\end{proof}

\section{Lower bounds on the extension complexity of Voronoi cells}
\label{sec:lower}

The aim of this section is to prove Theorem~\ref{thmLowerBound}. Inspired by Kannan's proof~\cite[Sec.~6]{kannan1987minkowski} of the $\np$-hardness of the closest vector problem, for every $0/1$-polytopes $P$ we are able to construct a lattice such that a face of its dual Voronoi cell projects onto $P$. To obtain a lattice of small dimension, $P$ needs to fulfill some extra condition.

\begin{lem} \label{lem38hduu}
    Let $H \subseteq \R^k$ be an affine subspace such that all vectors in $X := \{0,1\}^k \cap H$ have the same norm.
    There is a lattice $\lat$ with $\dim(\lat) \le \dim(H) + 1$ such that $\conv(X)$ is a linear projection of a face of $\vc(\lat)^\circ$.
\end{lem}
\begin{proof}
    Let $\alpha \ge 0$ be such that $\|x\| = \alpha$ for all $x \in X$.
    We may assume that $H$ is nonempty and that $\alpha > 0$, otherwise $\conv(X)$ is empty or consists of a single point, in which case the claim is trivial.
    Let $h \in H$ and let $L$ be the linear subspace such that $H = L + h$.
    Consider the lattice
    \[
        \lat := \left \{ z = (z',z'') \in \Z^k \times \alpha \Z : z' + \tfrac{1}{\alpha} z'' h \in L, \, \la \one, z' \ra + \alpha z'' = 0 \right \}
    \]
    and let $p := (\zero, -\alpha) \in \R^{k+1}$.
    We will show that
    \begin{equation} \label{eqskd99i}
        \cl(p,\lat) = \{\zero\} \cup \{(x,-\alpha) : x \in X\} =: U
    \end{equation}
    holds.
    The claim then follows from Lemma~\ref{lem:polarFace}.

    First note that $U \subseteq \lat$.
    Moreover, we have $\|p - \zero\| = \alpha$ and for each $x \in X$ we have $ \|p - (x,-\alpha)\| = \|x\| = \alpha $.
    Thus, in order to establish~\eqref{eqskd99i} it remains to show that every lattice point $z = (z',z'') \in \lat \setminus U$ satisfies $\alpha < \|p - z\|$.
    Equivalently, we have to show that every such point satisfies 
    \begin{equation}
    	\label{eq:toSatisfy}
    	f(z) := \|z'\|^2 + \|z'' + \alpha\|^2 > \alpha^2.
    \end{equation}
    This is clear if $z'' \notin \{0, -\alpha,-2\alpha\}$.
    If $z'' = 0$, then since $z \notin U$ we must have $z' \ne \zero$ and hence $f(z) = \|z'\|^2 + \alpha^2 > \alpha^2$.
    If $z'' = -\alpha$, then $z' \in H$ and $\la \one, z' \ra = \alpha^2$ hold.
    Since $z' \in \Z^k$, we obtain $f(z) = \|z'\|^2 \ge \la \one, z' \ra = \alpha^2 $ with equality only if $z' \in \{0,1\}^k$.
    However, in the latter case we would have $z' \in \{0,1\}^k \cap H = X$ and hence $z \in U$, a contradiction.
    Thus, we obtain~\eqref{eq:toSatisfy}.
    Finally, if $z'' = -2\alpha$, then $f(z) = \|z'\|^2 + \alpha^2$ and $\la \one, z' \ra = 2\alpha^2 > 0$, implying $z' \ne \zero$ and hence~\eqref{eq:toSatisfy} holds.
\end{proof}

While the previous lemma appears quite restrictive, the next lemma shows that we may apply it to a large class of $0/1$-polytopes.

\begin{lem} \label{lem8djsk9}
    Let $X = \{ x \in \{0,1\}^k : Ax \le b \}$, for some $A \in \R^{m \times k}$, $b \in \R^m$ such that $b - Ax \in \{0,1\}^m$, for all $x \in X$.
	There is a lattice $\lat$ of dimension at most $k+1$ such that $\conv(X)$ is the linear projection of a face of $\vc(\lat)^\circ$.
\end{lem}
\begin{proof}
    Consider the set
    \[
        X' := \{(x,x',s,s') \in \{0,1\}^{k+k+m+m} : Ax + s = b, x + x' = \one, s + s' = \one\}
    \]
    and observe that projecting $X'$ onto the first $k$ coordinates yields the set~$X$.
    Moreover, notice that every vector in $X'$ consists of exactly $k+m$ ones.
    In other words, the norm of every vector in $X'$ is $\sqrt{k+m}$ and hence, we may apply Lemma~\ref{lem38hduu} to obtain a lattice $\lat$ with dimension at most $k+1$ such that $\conv(X')$ is the linear projection of a face $F$ of $\vc(\lat)^\circ$.
    Since $\conv(X)$ is a linear projection of $\conv(X')$, we see that $\conv(X)$ is also a linear projection of $F$.
\end{proof}

\begin{proof}[Proof of Theorem~\ref{thmLowerBound}]
    We use a result of Göös, Jain \& Watson~\cite{GJW} that yields a family of $n$-node graphs $G$ such that the stable set polytope~$P_G$ of $G$ satisfies $\xc(P_G) = 2^{\Omega(n / \log n)}$. 
    Let $X \subseteq \{0,1\}^n$ denote the set of characteristic vectors of stable sets in $G$.
    Notice that
    \[
        X = \left \{ x \in \{0,1\}^n : x(i) + x(j) \le 1 \text{ for all } \{i,j\} \in E(G) \right \}.
    \]
    By Lemma~\ref{lem8djsk9}, there is a $d$-dimensional lattice $\lat$ with $d \le n+1$ such that $\conv(X)$ is a linear projection of a face $F$ of $\vc(\lat)^\circ$.
    We conclude
    \[
        \xc(\vc(\lat))
        = \xc(\vc(\lat)^\circ)
        \ge \xc(F)
        \ge \xc(\conv(X))
        = \xc(P_G)
        = 2^{\Omega(n / \log n)},
    \]
    and the claim follows since $d = \cO(n)$.
\end{proof}

\section{Spectrahedral lifts}
\label{sec:semidef}

A generalization of linear lifts of a polytope is the following. By $ \mathcal{S}^m $ we denote the set of all symmetric, real $ m \times m $ matrices. Moreover, by $ \mathcal{S}_+^m $ we denote the set of all those matrices in $ \mathcal{S}^m $ that are positive semidefinite (PSD). A \emph{spectrahedron} is a set containing all vectors $ x \in \R^n $ that fulfill conditions of the form $ M(x) \in \mathcal{S}_+^m $, where $ M: \R^n \rightarrow \mathcal{S}^m $ is an affine function. 
For a polytope $P$, the pair $ (Q, \pi) $, where $ Q \subseteq \R^n $ is a spectrahedron and $ \pi: \R^n \rightarrow \R^d $ is an affine map with $ \pi(Q) = P $, is called a \emph{(PSD) lift} of~$ P $.
The \emph{size} of this lift refers to the dimension of the matrix $ M(x) $. For $ Q = \{x \in \R^n : M(x) \in \mathcal{S}_+^m\} $ the size equals $ m $.
The \emph{semidefinite extension complexity} of $ P $, denoted by $ \sxc(P) $, is defined as the smallest size of any of its (PSD) lifts.

Given a polyhedron $ Q = \{ x \in \R^n : \la a_i , x \ra \le b_i, i \in [m] \} $, we can define $ M: \R^n \rightarrow \mathcal{S}^m $ via $ M(x)_{ii} = b_i - \la a_i , x \ra $, for all $ i \in [m] $ and $ M(x)_{ij} = 0 $ for $ i \ne j $ and hence $ Q = \{ x \in \R^n : M(x) \in \mathcal{S}_+^m \} $. This shows that every polyhedron is a spectrahedron and therefore
\[
\sxc(P) \le \xc(P).
\]
Hence, the upper bounds obtained in Section~\ref{sec:upperBound} also apply to the semidefinite case.

Furthermore, it is clear from the definition that for any polyhedron $ P $ and any affine map $ \pi $ we have that $\sxc(\pi(P)) \le \sxc(P)$.
Moreover, Lemma~\ref{lem:xcFace} and Lemma~\ref{lem:xc-dual} analogously hold in the semidefinite case, since Yannakakis' result on the nonnegative rank of a slack matrix was extended to (PSD) lifts in~\cite{fiorini2012linear, gouveia2013lifts}:
The semidefinite extension complexity $ P $ equals the \emph{PSD rank} of $ S $, which is the smallest dimension $ r $ for which there exist PSD matrices $ F_1, F_2, \dots, F_m \in \mathcal{S}_+^r $ and $ V_1, V_2, \dots, V_n \in \mathcal{S}_+^r $ such that $ S_{ij} = \la F_i, V_j \ra $ where the scalar product of two matrices is defined via $ \la A, B \ra = \sum_{i,j} A_{ij} B_{ij} $.

We obtain a superpolynomial lower bound on the semidefinite extension complexity of Voronoi cells of certain lattices using a lower bound of Lee, Raghavendra \& Steurer~\cite{lee2015lower} on semidefinite extension complexities of correlation polytopes.

\begin{thm}
	There exists a family of $d$-dimensional lattices $\lat$ such that $\sxc(\vc(\lat)) = 2^{\Omega(d^{1/13})}$.
\end{thm}
\begin{proof}
	In~\cite{lee2015lower} it is proven that the semidefinite extension complexity of the \emph{correlation polytope} 
	\[
	P_n = \conv \left\{ xx^\intercal : x \in \{0,1\}^n  \right\}
	\]
	is bounded from below by $2^{\Omega(n^{2/13})}$.
	Notice that
	\begin{align*}
		P_n = \conv \Big\{ Y \in \{0,1\}^{n \times n} : &\ Y_{ij} \le Y_{ii}\ ,\ Y_{ij} \le Y_{jj} \text{ and } Y_{ii} + Y_{jj} - 1 \le Y_{ij}, \\ &\text{ for all } i,j \in [n] \text{ with } i \neq j \Big\}.
	\end{align*}
	Hence, the correlation polytope can be written as the convex hull of binary vectors $ Y \in \{0,1\}^{n \times n} $ satisfying linear inequalities whose slacks only have values in $\{0,1\}$.
	Therefore, by Lemma~\ref{lem8djsk9} there is a lattice of dimension $d$ where $d \le n^2 + 1 = \Theta(n^2)$ such that $P_n$ is a linear projection of a face $F$ of $\vc(\lat)^\circ$.
	Analogously to the proof of Theorem~\ref{thmLowerBound} for the linear extension complexity, we conclude
	\[
	\sxc(\vc(\lat))
	= \sxc(\vc(\lat)^\circ)
	\ge \sxc(F)
	\ge \sxc(P_n)
	= 2^{\Omega(n^{2/13})},
	\]
	and the claim follows since $n = \Omega(\sqrt{d})$.
\end{proof}

\section{Open questions}
\label{sec:openquestions}

We conclude our investigations of the extension complexity of Voronoi cells of lattices with a collection of some open problems that naturally arise from our studies and which we find interesting to pursue in future research.

In view of Theorem~\ref{thmLowerBound} a natural question is whether the logarithmic term in the lower bound $2^{\Omega(d / \log d)}$ on the extension complexity of certain Voronoi cells can be removed:

\begin{ques}
\label{qu:xc-2tod}
Does there exist a family of $d$-dimensional lattices $\lat$ such that $\xc(\vc(\lat)) = 2^{\Omega(d)}$?
\end{ques}

We remark that our bound relies on a lower bound by G\"o\"os, Jain \& Watson~\cite{GJW} on extension complexities of stable set polytopes, which meet the criteria of Lemma~\ref{lem8djsk9}.
It is known that there \emph{exist} $d$-dimensional $0/1$-polytopes with extension complexity $2^{\Omega(d)}$, see~\cite{rothvoss2013some}.
However, no explicit construction of such polytopes is known and so it is unclear how to transform such polytopes in order to apply Lemma~\ref{lem38hduu} efficiently.

Comparing the superpolynomial bound in Theorem~\ref{thmLowerBound} with the polynomial upper bounds for certain classes of lattices in Section~\ref{sec:upperBound}, the question arises what we can expect the extension complexity of the Voronoi cell of a \emph{generic} lattice to be.

\begin{ques}
What is $\xc(\vc(\lat))$ for a ``random'' $d$-dimensional lattice $\lat$?
\end{ques}

Of course, this requires a suitable notion of a random lattice.
Our question refers to interesting examples such as Siegel's measure~\cite{siegel1945mean} or uniform distributions over integer lattices of a fixed determinant, see Goldstein \& Mayer~\cite{goldstein2003equidistribution}.

In Theorem~\ref{thmLowerBound} we have shown that \emph{exactly} describing a Voronoi cell of a lattice may require superpolynomial-size extended formulations.
It would be interesting to understand how this situation changes if we allow approximations instead of exact descriptions, in particular in view of various results on the complexity of the approximate closest vector problem, see, e.g., Aharonov \& Regev~\cite{aharonov2005lattice}.
To this end, for $\alpha \ge 1$ we say that a polytope~$Q$ is an \emph{$\alpha$-approximation} of a polytope $P$, if $P \subseteq Q \subseteq \alpha P$.

\begin{ques}
What can be said about extension complexities of $\alpha$-approximations of Voronoi cells of lattices?
\end{ques}

We have seen in Theorem~\ref{thm:xc-root-lattice} that not only the root lattices but also their dual lattices have polynomial extension complexity.
Is that a general phenomenon?

\begin{ques}
Given a $d$-dimensional lattice $\lat$, is there a polynomial relationship between $\xc(\vc(\lat))$ and $\xc(\vc(\lat^\star))$?
\end{ques}

Given that in view of Theorem~\ref{thm:zonotopal-lattices} zonotopal lattices admit lifts with quadratically many facets, and the fact that the closest vector problem on such lattices can be solved in polynomial time (see~\cite{mccormickpeisscheidweilervallentin2020apolynomial}), one might expect that small-sized lifts of the corresponding Voronoi cells can actually be constructed explicitly.

\begin{ques}
Given a basis of a $d$-dimensional zonotopal lattice $\lat$, is it possible to construct an explicit lift of $\vc(\lat)$ with polynomially many facets in polynomial time?
\end{ques}

Note that our arguments leading to Theorem~\ref{thm:zonotopal-lattices} are not constructive.

\section*{Acknowledgements}
The third author was supported by the Deutsche Forschungsgemein\-schaft (DFG, German Research Foundation), project number 451026932.
We like to thank Gennadiy Averkov, Daniel Dadush, and Christoph Hunkenschröder for valuable discussions on the topic.

\bibliographystyle{plain}
\bibliography{references}

\end{document}